\documentclass[10pt, a4paper, oneside, reqno]{amsart}

\makeatletter
\def\@settitle{\begin{center}%
		\baselineskip14\p@\relax
		\normalfont\LARGE\scshape\bfseries
		\@title
	\end{center}%
}
\makeatother

\makeatletter

\def\subsection{\@startsection{subsection}{2}%
	\z@{.5\linespacing\@plus.7\linespacing}{.5\linespacing}%
	{\normalfont\bfseries}}

\def\subsubsection{\@startsection{subsubsection}{3}%
	\z@{.5\linespacing\@plus.7\linespacing}{.5\linespacing}%
	{\normalfont\itshape}}

\usepackage[usenames, dvipsnames]{color}
\definecolor{darkblue}{rgb}{0.0, 0.0, 0.45}

\usepackage[colorlinks	= true,
			raiselinks	= true,
			linkcolor	= darkblue, 
			citecolor	= Mahogany,
			urlcolor	= ForestGreen,
			pdfauthor	= {},
			pdftitle	= {},
			pdfkeywords	= {},
			pdfsubject	= {},
			plainpages	= false]{hyperref}

\usepackage{dsfont,amssymb,amsmath,graphicx} 
\usepackage{amsfonts,dsfont,mathtools,mathrsfs,amsthm} 
\mathtoolsset{showonlyrefs=true}
\usepackage[amssymb, thickqspace]{SIunits}
\usepackage{fancyhdr,setspace}

\usepackage[numbers,compress]{natbib}
\usepackage{nicefrac}       
\usepackage{microtype}      
\usepackage{tabto}

\usepackage{enumerate,enumitem}
 \usepackage{caption}

\usepackage{bbm}

\allowdisplaybreaks
\date{\today}

\graphicspath{{images/}}

\usepackage{url}            

\usepackage{algorithm}
\usepackage{algorithmic}

 \usepackage[foot]{amsaddr}

\usepackage{subfig}

\usepackage{algorithmic}
{

{
		\theoremstyle{plain}
		\newtheorem{theorem}{Theorem}
	}
	{
		\theoremstyle{plain}
		\newtheorem{assumption}{Assumption}
	}
	{
	\theoremstyle{plain}
	
}
	{
		\theoremstyle{plain}
		\newtheorem{definition}{Definition}
	}
	{
		\theoremstyle{plain}
		
	}
{
		\theoremstyle{plain}
		\newtheorem{remark}{Remark}
	}
{
		\theoremstyle{plain}
		
	}

\title[]{Maintenance scheduling of manufacturing systems based on optimal price of the network}

\author[1]{Pegah Rokhforoz$^1$}
\author[2]{Olga Fink$^{2}$}
\address[1]{Chair of Intelligent Maintenance Systems, ETH Zurich, Switzerland, and School of Electrical and Computer Engineering, University of Tehran, Tehran, Iran.}
\address[2]{Chair of Intelligent Maintenance Systems, ETH Zurich, Switzerland. Corresponding author, email address:ofink@ethz.ch.}

\begin{document}

\maketitle

\begin{abstract}
Goods can exhibit positive externalities impacting decisions of customers in socials networks. Suppliers can integrate these externalities in their pricing strategies to increase their revenue. Besides optimizing the  prize, suppliers also have to consider their production and maintenance costs. Predictive maintenance has the potential to reduce the maintenance costs and improve the system availability. To address the joint optimization of pricing with network externalities and predictive maintenance scheduling based on the condition of the system, we propose a bi-level optimization solution based on game theory. In the first level, the manufacturing company decides about the predictive maintenance scheduling of the units and the price of the goods. 
In the second level, the customers decide about their consumption using an optimization approach in which the objective function depends on their consumption, the consumption levels of other customers who are connected through the graph, and the price of the network which is determined by the supplier. To solve the problem, we propose the leader-multiple-followers game where the supplier as a leader predicts the strategies of the followers. Then, customers as the followers obtain their strategies based on the leader's and other followers' strategies. We demonstrate the effectiveness of our proposed method on a simulated case study. The results demonstrate that knowledge of the social network graph results in an increased revenue compared to the case when the underlying social network graph is not known. Moreover, the results demonstrate that obtaining the predictive maintenance scheduling based on the proposed optimization approach leads to an increased profit compared to the baseline decision-making (perform maintenance at the degradation limit). 
\end{abstract}

\section{introduction}
Social networks can have a significant influence on the economic system \cite{jadbabaie2019optimal,candogan2012optimal}. In a social network, users connect and communicate with each other. Hence, the decision of a user has the potential to affect the decisions of other users, particularly for goods that exhibit externalities. In such cases, users receive different extents of externalities from the consumption of other users to whom they are connected in their social network. The information about the consumption behavior of users spreads in the network. If suppliers know the graph network and the consumption behaviour of the users, they can try to integrate this information in their optimization. Aiming to maximize their revenue, suppliers determine the price of the good for each user such that their revenue is maximized.

Besides considering the prizing of the manufactured goods, manufacturers also need to consider the production and the maintenance costs. For example, deterioration of manufacturing units can significantly decrease the production quality \cite{lu2019quality}. The quality of the production depends, among others, on the maintenance strategy of the manufacturing equipment to prevent production disruptions \cite{Sarkar_2011}, \cite{erozan2019fuzzy}. In fact, the goal of maintenance scheduling is to maximize the availability of the system, while fulfilling the demand of the users or customers. Traditionally, preventive maintenance is scheduled regularly to restore the units to their healthy states. However, preventive maintenance does not consider the system's condition and the dynamic changes \cite{Omshi_2020}. Predictive maintenance has the potential to significantly reduce the maintenance costs. Furthermore, maintenance scheduling does typically not consider the consumption behaviour of its customers which influences the prizing but also the produced amount which in turn influences the requirements on the availability of the system but also the degradation of the system and the demand for maintenance actions. Hence, in order to maximize the availability and reward function, it is advisable for a supplier  to schedule the maintenance of its manufacturing units using the underlying graph of the social network of the users  and the users' consumption function.

In this paper, we address the maintenance scheduling problem of a manufacturing systems of goods with positive externalities. In this context, the goal of the manufacturer is to find the optimal price of the goods and the optimal maintenance scheduling for the manufacturing units such that its revenue and availability are maximized. In order to obtain the maintenance scheduling, we consider the framework of predictive maintenance in which the remaining useful lifetime (RUL) of the units and the associated uncertainties can be predicted. We define the optimization model for the manufacturer which comprises two main elements: 1) the revenue that can be obtained by selling the produced goods and 2) maintenance cost of the manufacturing units. To model the positive externality effects, we assume that users or customers are connected through a graph and that they choose their consumption level of the good based on their utility function and the price of the good. Hence, we face a bi-level optimization problem with the two levels representing the  supplier and the users. In order to solve the optimization problem, we propose to formulate the problem as a leader-follower game between the supplier and the users. In this game, the supplier is the leader and can predict the strategies (consumption function) of the followers. Then, based on this prediction, it determines the maintenance scheduling and optimal price of the network and sets the price for each of the users. We investigate the efficiency of the proposed approach based on a case study. 

To the best of our knowledge, this is the first research that obtains the maintenance scheduling of different units in the manufacturing system by considering the price of the network and the positive externality effects of the goods using a leader-follower game among the manufacturer and the customers.

The rest of the paper is organized as follows. The review of related work is presented in Section \ref{sec:relate}. The preliminaries on predicting the remaining useful lifetime (RUL) and the leader-multiple-followers game are introduced in Section \ref{sec:preliminary}. The supplier's and customers' objective functions are formulated in Section \ref{sec:problem}. The solution method based of the leader-follower game is proposed in Section \ref{sec:solution}. The simulation results are presented in Section \ref{sec:simulation}. The conclusion remarks are made in Section \ref{sec:conclusion}.

\section{Related work}
\label{sec:relate}
\textbf{Network price of goods with positive externalities.}
Several approaches have been proposed to determine the optimal price in networks with positive externalities. The authors of \cite{candogan2012optimal} study the optimal pricing strategies of a supplier of goods with positive externalities assuming that consumers are connected by a social network graph. Their usage levels depend directly on the usage of their neighbors within a social network \cite{candogan2012optimal}. Since the users' consumption behaviour affects that of other users, it has been modeled based on the game theory concept \cite{wu2002optimal}, \cite{bramoulle2014strategic}.
The optimal contracting between the seller and buyers is proposed in \cite{wu2002optimal}. The authors develop the solution method using the Stackelberg game where the seller acts as a leader and the buyers are followers. The strategic interaction and networks have been addressed using game theory in \cite{bramoulle2014strategic}. In this study, the decisions of some suppliers affect those of other suppliers in the network. 

In all of these works, the manufacturers solely consider the pricing of the goods and do not take the impact of the production on the operation and maintenance costs into account. In fact, sudden failures causing an interruption of the production can severally impact the ability of the supplier to respond to the demand. Furthermore, not considering the interruptions caused by planned maintenance in the pricing will result in sub-optimal prices and revenues. To maximize their revenue, the suppliers should strive to perform maintenance on their manufacturing units when the price of the network is low and they sell less.

\textbf{Maintenance scheduling.} 
Maintenance scheduling with the aim of increasing the system reliability and availability has been addressed in many applications such as electrical market \cite{rokhforoz2021multi}, \cite{volkanovski2008genetic}, network transportation \cite{Amiri_2006designing}, \cite{feng2018two}, and manufacturing system \cite{yang2008maintenance}, \cite{celen2020integrated}.

Several studies have proposed different approaches for maintenance scheduling and optimal management of manufacturing systems. The authors of \cite{yang2008maintenance} propose a method based on a Genetic algorithm to obtain the maintenance scheduling of manufacturing equipment using the predicted level of degradation. The maintenance scheduling for manufacturing systems proposed in \cite{lu2017opportunistic} considers the product quality and the maintenance cost including the repair and preventive maintenance cost. Dynamic maintenance scheduling which obtains the optimal predictive maintenance using dynamic programming method is developed in \cite{Ghasemi_2007}, \cite{Ghasemi_2008}, and  \cite{kroning2013dynamic}. In these studies, the deterioration process of the units is formulated as a Markov decision process. An optimization-based approach for dynamic maintenance scheduling of different units which considers the operating conditions is proposed in \cite{celen2020integrated}. In all of these studies, the maintenance scheduling is obtained only based on the degradation level of different units without considering the price of the goods and the effect of maintenance on the revenue of the manufacturer.

\paragraph*{\textbf{Notation}:} Given $D\in{\mathbb{R}^{n\times{n}}}$, $D^{-1}$ denotes the inverse of $D$.  We define the column augmentation of $Z_{n}(t)$ for $t=1,\cdots,T$, as ${Z}_{n}=\textbf{col}(Z_{n}(1),\cdots,Z_{n}(T)):=[Z_{n}(1),\cdots,Z_{n}(T)]$.

\section{Preliminaries}
\label{sec:preliminary}
\subsection{Remaining useful life (RUL) and predictive maintenance}
The remaining useful life (RUL) is defined as the amount of time that an asset will continue to satisfy its desired operating conditions \cite{lee2014prognostics}. Predicting RUL is one of the core tasks in predictive maintenance applications. Predicting the RUL enables on the one hand to perform maintenance before the failure occurs and by that also to improve the availability of the system while on the other hand, the lifetime of the system can be fully exploited resulting in less frequent replacements or repairs. Several methods have been proposed in the literature to estimate RUL. These methods can be categorized in three categories: 1) model-based approaches, where the physics principle models of the degradation of the asset are applied \cite{liao2021remaining}; 2) data-driven methods, in which the RUL is estimated based on the condition monitoring data only \cite{si2012remaining}; 3) knowledge-based approaches that depend on the domain knowledge of an expert \cite{djeziri2018hybrid}. RUL estimations are always subject to uncertainties. While not all of the RUL prediction methods also estimate the uncertainty of the predictions, uncertainty quantification is an integral part of decision support systems and is particularly desirable by domain experts using such systems. For many systems with multiple-units that are jointly fulfilling the demand, performing predictive maintenance at the end of life of the system may not the optimal. Since, the maintenance decision of each unit also depends on the decision of other units and the production requirements. Hence,  predictive maintenance scheduling is pivotal for fulfilling the demand while minimizing the maintenance costs.
In our problem formulation, we assume that we are able to predict the RUL and can quantify the associated uncertainty. We assume that the distribution function of RUL is given. The manufacturer should, therefore,  take into account the uncertainty of RUL and schedule the maintenance based on the production requirements and the associated costs of the unavailability, while considering that a part of the lifetime of the unit may be lost when replacing the unit before the end of life. 

\subsection{Leader-multiple-followers game}
Let us consider N-players, $i=1,\cdots,N$, as followers and one player as a leader. Let $x_{i}$ denote the follower $i$'s strategy and $y$ the leader's strategy, respectively. Let $U^{f}_{i}(x_{i},x_{-i},y)$ denote the utility function of the follower $i$ where $x_{-i}$ is the strategy of all followers except follower $i$. Let $U^{l}(y,x)$, in which $x=\textbf{col}(x_{1},\cdots,x_{n})$, denote the utility function of the leader. The equilibrium point of the leader-multiple-followers game can be defined as follows:
\begin{definition}
\label{def1}
Let $(x^*,y^*)$ be the equilibrium point among the leader and the followers, then,
\begin{equation}
\begin{aligned}
    &U^f_{i}(x^*_{i},x^*_{-i},y^*)\geq{ U^f_{i}(x_{i},x^*_{-i},y^*)},\quad i=1,\cdots,N,\\
    &U^{l}(y^*,x^*)\geq{U^{l}(y,x^*)}
    \end{aligned}
\end{equation}
\end{definition}
In the leader-follower game, the leader predicts the followers' strategies and implements his strategy first, while the followers react to the leader's decisions.

\section{Problem formulation}
\label{sec:problem}
In this section, we propose the optimization model for obtaining the maintenance scheduling of the units in the manufacturing system based on the optimal price of the network. 

 Let us define $\mathcal{J}={\{1,\cdots,J}\}$ as the set of units in the manufacturing system that are manufacturing the good that has positive externalities to the customers. Let us define $\mathcal{N}={1,\cdots,N}$ as the set of customers who are connected through the graph matrix $W$ in the social network. The $il$ element of matrix $W$ denoted by $w_{il}\geq{0}$ represents the strength of the connection between the two customers in the network and represents concurrently also the influence of customer $l$ on customer $i$. The manufacturer aims to obtain the optimal price of the network while maintaining a high  availability of the manufacturing units. We assume that the manufacturer has implemented a predictive maintenance strategy to improve the availability of the system and is able to predict the RUL of the systems and the associated uncertainty. Predicting the RUL enables manufacturer to schedule the maintenance before the end of life ant thereby, prevent unscheduled down times while fully exploiting the lifetimes of the units. The manufacturers aims on the one hand to maintain a high availability of the units and exploit the useful lifetime of the units as much as possible, while performing the maintenance at points in time when it has the least possible impact on its revenue. The manufacturer is, therefore, seeking to obtain the maintenance scheduling of manufacturing units for the decision horizon time $\mathcal{T}={\{1,\cdots,T}\}$ while maximizing the revenue that can be obtained by selling the products to the customers. In the following, we explain the proposed optimization model.

\subsection{Manufacturer's model: maintaining and pricing}
In the following, we propose the manufacturer's optimization problem which concurrently optimizes the price of the network and the maintenance scheduling of the units.

\textit{Maintaining:} In this study, we assume that the deterioration state of the unit increases with time. The maintenance restores the state of the unit to its initial condition (as good as new). Hence, we model the deterioration and the restoration as follows:
\begin{equation}
s_{j}(t+1)=(1-x_{j}(t))s_{j}(t)+1,
    \label{eq:rul}
\end{equation}
where $s_{j}(t)$ and $s_{j}(t-1)$ denote the deterioration state of unit $j$ at instants $t$ and $t-1$, respectively. $x_{j}(t)$ is a binary variable and $x_{j}(t)=1$ denotes that unit $j$ performs maintenance at instant $t$. 

When the deterioration state reaches the degradation threshold, it reaches the end of life which is equivalent to the remaining useful life (RUL) reaching zero. If the unit has not been maintained before the end of life, it will fail. As discussed above, the prediction of the RUL is always associated with uncertainties. It can, therefore, be considered as a random variable with a known distribution function. We assume that the predicted RUL and the associated uncertainty are given. To avoid failure, the units should perform maintenance before the end of life. To model this aim, we apply the chance constraint problem as follows:
\begin{equation}
{\mathbb{P}}{\Big(S_{2,j}\in\mathcal{A}|\quad s_{j}(t)-{S_{2,j}}\leq{0}\Big)}\geq{1-\alpha},
\label{eq:chance constraint}
\end{equation}
where $\mathbb{P}$ is a probability measure defined over $\mathcal{A}$. $S_{2,j}$ is the degradation threshold of the deterioration state of unit $j$. $s_{j}(t)$ which satisfies the chance constraint \eqref{eq:chance constraint} can be considered as the $\alpha$-level feasible solution. 

Another constraint of this problem is that the total consumption of all customers should be less than the total production of all units. We assume that the demand cannot be postponed and is satisfied at all times. Furthermore, we assume that the good is perishable and cannot be stored to cover the demand during the maintenance down time.  Moreover, we assume that when the unit performs maintenance, its production is zero. Hence, we can model this constraint as follows:

\begin{equation}
    \sum\limits_{i\in\mathcal{N}}q_{i}(t)\leq{\sum\limits_{j\in\mathcal{J}}(1-x_{j}(t))q_{j,max}},\\
    \label{eq:capacity}
\end{equation}
where $q_{i}(t)$ is the amount of the consumption of customer $i$ at instant $t$. $q_{j,\max}$ is the maximum production output by manufacturing unit $j$.

The supplier seeks to obtain the maintenance scheduling which satisfies \eqref{eq:chance constraint} and \eqref{eq:capacity}, while concurrently minimizing the maintenance cost of the units which is modeled as follows:

\begin{equation}
  C=\sum\limits_{j\in\mathcal{J}}\sum\limits_{t\in\mathcal{T}}c_{j}x_{j}(t),\\
  \label{eq:cost}
    \end{equation}
 where $c_{j}$ is the cost of unit $j$.   
 
\textit{Pricing:} The supplier aims to obtain the price of the network which maximizes the following reward:

\begin{equation}
    R=\sum\limits_{i\in\mathcal{N}}\sum\limits_{t\in\mathcal{T}}\phi_{i}(t){q_{i}(t)},\\
    \label{eq:reward}
\end{equation}
where $\phi_{i}(t)\geq{0}$ is the price of the consumed good of customer $i$ at instant $t$. $q_{i}(t)\geq{0}$ is the consumption of customer $i$ at instant $t$.

\textit{Supplier's objective function:} Using \eqref{eq:cost} and \eqref{eq:reward}, the utility function of the supplier can be expressed as follows:
\begin{equation}
    U=\sum\limits_{i\in\mathcal{N}}\sum\limits_{t\in\mathcal{T}}\phi_{i}(t){q_{i}(t)}-\sum\limits_{j\in\mathcal{J}}\sum\limits_{t\in\mathcal{T}}c_{j}x_{j}(t),
    \label{eq:utility}
\end{equation}

Hence, by applying \eqref{eq:utility}, and by considering constraints \eqref{eq:rul}, \eqref{eq:chance constraint}, and \eqref{eq:capacity}, we can formulate the supplier's objective function as follows:

\begin{equation}
\begin{aligned}
    &\max\limits_{{\phi},{x},{s}}\sum\limits_{i\in\mathcal{N}}\sum\limits_{t\in\mathcal{T}}\phi_{i}(t){q_{i}(t)}-\sum\limits_{j\in\mathcal{J}}\sum\limits_{t\in\mathcal{T}}c_{j}x_{j}(t)\\
    &\text{S.b.}\quad \text{C}_{1}:\quad s_{j}(t+1)=(1-x_{j}(t))s_{j}(t)+1,\\
    &\quad \qquad \text{C}_{2}:\quad{\mathbb{P}}{\Big(S_{2,j}\in\mathcal{A}|\quad s_{j}(t)-{S_{2,j}}\leq{0}\Big)}\geq{1-\alpha},\\
    &\quad \qquad \text{C}_{3}:\quad\sum\limits_{i\in\mathcal{N}}q_{i}(t)\leq{\sum\limits_{j\in\mathcal{J}}(1-x_{j}(t))q_{j,max}},\\
    &\quad \qquad \text{C}_{4}:\quad q_{i}(t)\geq{0},\quad \phi_{i}(t)\geq{0},
    \end{aligned}
    \label{eq:main problem}
\end{equation}
where ${\phi}=\textbf{Col}{\{\phi_{1},\cdots,\phi_{N}}\}$, ${\phi}_{i}=\textbf{Col}{\{\phi_{i}(1),\cdots,\phi_{i}(T)}\}$, ${x}=\textbf{Col}{\{x_{1},\cdots,x_{N}}\}$, ${x}_{i}=\textbf{Col}{\{x_{i}(1),\cdots,x_{i}(T)}\}$,
${s}=\textbf{Col}{\{s_{1},\cdots,s_{N}}\}$, ${s}_{i}=\textbf{Col}{\{s_{i}(1),\cdots,s_{i}(T)}\}$, $i\in\mathcal{N}$.

The problem \eqref{eq:main problem} is a nonlinear mixed-integer programming problem due to the constraints $\text{C}_{1}$ and $\text{C}_{2}$. We convert \eqref{eq:main problem} to the mixed-integer linear programming by using the big-M method \cite{fortuny1981representation} and scenario-based approach \cite{margellos2014road}. In the following, we explain these two approaches.

First, let us introduce a new variable $y_{j}(t)=x_{j}(t)s_{j}(t)$. By using the big-$M$ method, we can convert constraint $\text{C}_{1}$ to the following mixed-integer linear constraints:

\begin{equation}
    \begin{aligned}
    &s_{j}(t+1)=s_{j}(t)-y_{j}(t)+1,\\
    &y_{j}(t)\geq{s_{j}(t)-(1-x_{j}(t))M},\\
    &y_{j}(t)\leq{s_{j}(t)+(1-x_{j}(t))M},\\
     &{0}\leq{y_{j}(t)}\leq{x_{j}(t)}M.
    \label{eq:big m}
\end{aligned}
\end{equation}

Now, we can substitute the chance constraint \eqref{eq:chance constraint} with the finite number of constraints using the scenario based approach. Each constraint corresponds to a different realization of $S_{2,j}^{k}$, $k=1,\cdots,K$, for the uncertain parameter $S_{2,j}$. Thus, we have:

\begin{equation}
    s_{j}(t)-S_{2,j}^{k}\leq{0},\quad k=1,\cdots,K.
    \label{eq:scenario}
\end{equation}

\begin{remark}
The number of the scenarios $K$ should be chosen sufficiently large such that the feasible solution of \eqref{eq:scenario} is an $\alpha$-level feasible solution of \eqref{eq:chance constraint}.
\end{remark}

Hence, by using \eqref{eq:big m} and \eqref{eq:scenario}, the optimization problem \eqref{eq:main problem} leads to the following mixed-integer linear programming problem:

\begin{equation}
\begin{aligned}
    &\max\limits_{{\phi},{x},{s},{y}}\sum\limits_{i\in\mathcal{N}}\sum\limits_{h\in\mathcal{H}}\phi_{i}(t){q_{i}(t)}-\sum\limits_{j\in\mathcal{J}}\sum\limits_{t\in\mathcal{T}}c_{j}x_{j}(t)\\
    &\text{S.b.}\quad \text{C}'_{1}:\quad   s_{j}(t+1)=s_{j}(t)-y_{j}(t)+1,\\
    &\quad \qquad \text{C}'_{2}:\quad y_{j}(t)\geq{s_{j}(t)-(1-x_{j}(t))M},\\
    &\quad \qquad \text{C}'_{3}:\quad y_{j}(t)\leq{s_{j}(t)+(1-x_{j}(t))M},\\
     &\quad \qquad \text{C}'_{4}:{0}\leq{y_{j}(t)}\leq{x_{j}(t)}M.\\
    &\quad \qquad \text{C}'_{5}:\quad    s_{j}(t)-S_{2,j}^{k}\leq{0},\quad k=1,\cdots,K.
\\
    &\quad \qquad \text{C}'_{6}:\quad\sum\limits_{i\in\mathcal{N}}q_{i}(t)\leq{\sum\limits_{j\in\mathcal{J}}(1-x_{j}(t))q_{j,max}},\\
    &\quad \qquad \text{C}'_{7}:\quad q_{i}(t)\geq{0},\quad \phi_{i}(t)\geq{0},
    \end{aligned}
    \label{eq:main problem1}
\end{equation}
where ${y}=\textbf{Col}{\{y_{1},\cdots,y_{N}}\}$, ${y}_{i}=\textbf{Col}{\{y_{i}(1),\cdots,y_{i}(T)}\}$.

The consumption $q_{i}(t)$ is obtained by the customers' objective function, which is explained in the following.

\subsection{Customers' model: consumption}

The customers seek to obtain the optimal consumption which maximizes their reward function. Let us model the customer $i$ objective function as follows:

\begin{equation}
    \max\limits_{{q}_{i}}\sum\limits_{t\in\mathcal{T}}\big(v_{i}(q_{i}(t))+\sum\limits_{l\in\mathcal{N}}w_{il}q_{l}(t)q_{i}(t)-\phi_{i}(t){q}_{i}(t)\big),
    \label{eq:customer}
\end{equation}
where ${q}_{i}=\textbf{Col}{\{q_{i}(1),\cdots,q_{i}(T)}\}$. The objective function consists of three parts. The first term denotes the reward that customer $i$ can obtain by consuming the good. Inspired by the studies in \cite{candogan2012optimal} and \cite{jadbabaie2019optimal}, we model the first reward term as follows:
\begin{equation}
    v_{i}(q_{i}(t))=-\frac{1}{2}a_{i}q_{i}^{2}(t)+b_{i}q_{i}(t),
\end{equation}
where $a_{i}\geq{0}$ and $b_{i}$ are the customer's $i$ model. The second term expresses the effect of other customers which are connected through the social network graph to customer $i$. This term implies that the consumption $q_{k}(t)$ of customer $k$, who is connected to customer $i$ through the social network graph with the weight $w_{ik}$, influences the objective function of customer $i$. The third term formulates the price that customer $i$ is willing to pay to the supplier for obtaining the consumed good $q_{i}(t)$.

\section{Proposed solution: leader and multiple-followers game}
\label{sec:solution}
In this section, we propose the solution method using the concept of leader-multiple-followers game \cite{hu2013existence}, \cite{cruz1978leader} for solving the bi-level optimization problem (\eqref{eq:main problem1} and \eqref{eq:customer}). In this problem, we consider the customers as the followers who are seeking to obtain their consumption and the supplier as the leader who is responsible to obtain the price of the network and maintenance scheduling of its manufacturing units. The schematic of the proposed framework is shown in Figure \ref{schematic}.

\begin{figure}
	\centering
	\includegraphics[width=5in]{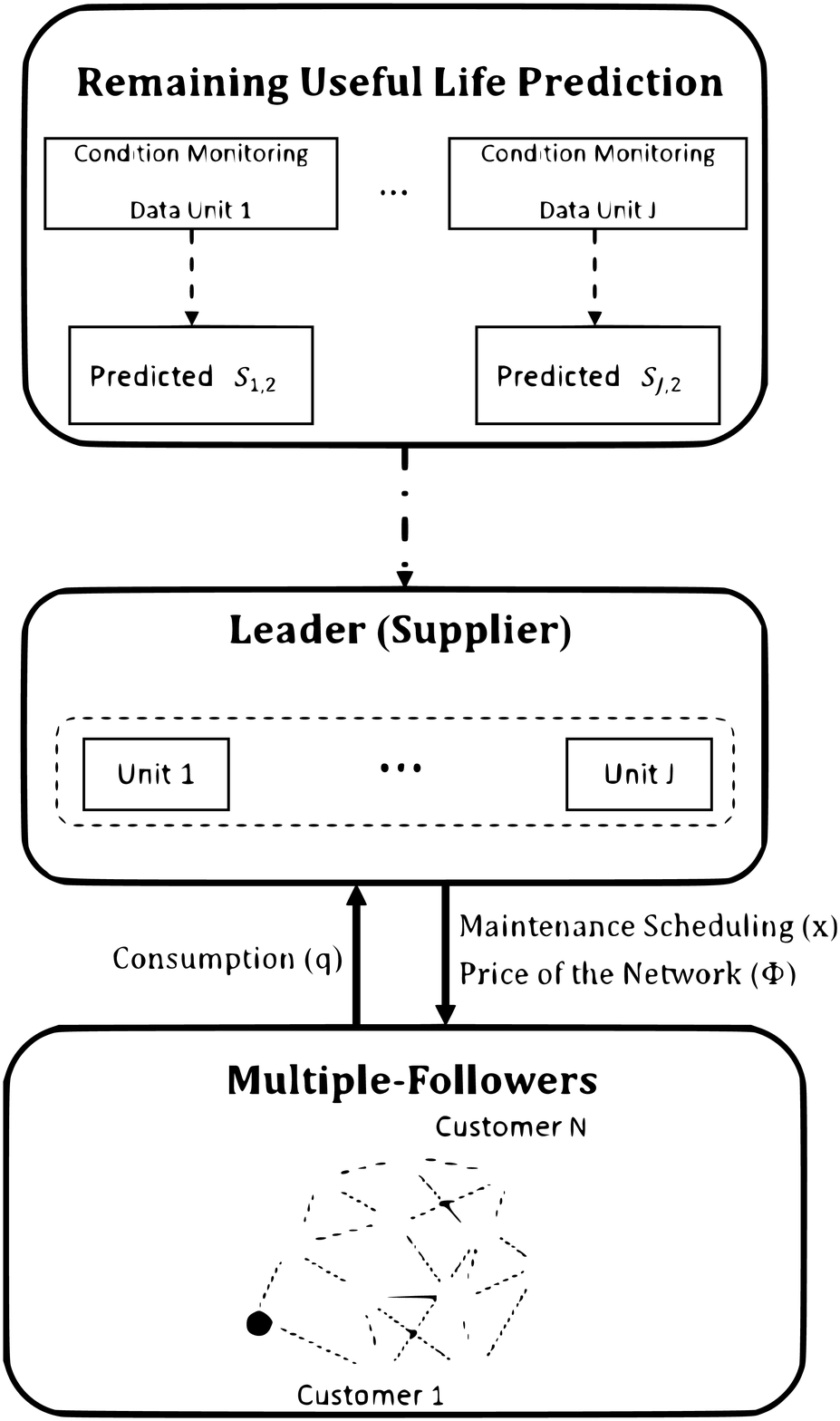}
	\caption{Schematic of the leader-multiple-followers game for obtaining the maintenance scheduling based on the price of the network.}
	\label{schematic}
\end{figure}

\subsection{Consumption equilibrium of multiple-follower}
In the proposed solution, we assume that the customers are the followers of the supplier. According to \eqref{eq:customer}, since the customers' objective function depends on the strategies of other customers, we face a game theory problem among the customers which can be described as $G=(\mathcal{N},\mathcal{Q},\mathcal{U})$:

1) $\mathcal{N}$ denotes the set of customers  as the players. 

2) $\mathcal{Q}=\prod\limits_{i\in\mathcal{N}}q_{i}$ denotes the strategy space of the players. 

3) $\mathcal{U}={\{U_{1},\cdots,U_{N}}\}$ is the set of utilities, where the utility of player $i$ is defined as follows:

\begin{equation}
    U_{i}(q_{i},q_{-i},\phi_{i})=    \sum\limits_{t\in\mathcal{T}}\big(v_{i}(q_{i}(t))+\sum\limits_{l\in\mathcal{N}}w_{il}q_{l}(t)q_{i}(t)-\phi_{i}(t){q}_{i}(t)\big)
\end{equation}
where $q_{-i}={\{q_{1},\cdots,q_{i-1},q_{i+1},\cdots,q_{N}}\}$ are the strategies of all players except player $i$. 

The Nash equilibrium (NE) is one appropriate output solution of the game. Using  Definition \eqref{def1}, the NE among customers can be defined as follows:
\begin{equation}
    U_{i}(q^*_{i},q^*_{-i},\phi^*_{i})\geq{U_{i}(q_{i},q^*_{-i},\phi^*_{i})}
\end{equation}

In order to find the NE of the game among customers, let us define the following assumption:

\begin{assumption}
\label{assumption 1}
$a_{i}\geq{\sum\limits_{l\in\mathcal{N}}w_{il}}$, $i\in\mathcal{N}$.
\end{assumption}

\begin{theorem}
Under Assumption \eqref{assumption 1}, the game $G=(\mathcal{N},\mathcal{Q},\mathcal{U})$ has a unique NE which is defined as follows:
\begin{equation}
    q^*(t)=(A-W)^{-1}(B-\Phi^*(t)),
    \label{eq:NE}
\end{equation}
where $A=\mathop{\mathrm{diag}}(a_{1},\cdots,a_{N})$, $W=\begin{pmatrix}w_{11}&\cdots&w_{1N}\\\vdots&\cdots&\vdots\\w_{N1}&\cdots&w_{NN}\end{pmatrix}$, $B=\begin{pmatrix}b_{1}\\\vdots\\b_{N}\end{pmatrix}$, $\Phi^*=\begin{pmatrix}\phi^*_{1}\\\vdots\\\phi^*_{N}\end{pmatrix}$
\end{theorem}
\begin{proof}
The best response of player $i$ is as follows:
\begin{equation}
    q_{i}(t)=\beta_{i}(q_{-i}(t))=\frac{b_{i}-\phi^*_{i}}{a_{i}}+\frac{\sum\limits_{l\in\mathcal{N}}w_{il}q_{l}(t)}{a_{i}},\quad {i\in\mathcal{N}},
    \label{eq:best response}
\end{equation}
Hence, \eqref{eq:best response} can be written in the matrix form as follows:
\begin{equation}
    Aq({t})=B-\Phi^*(t)+Wq(t)
\end{equation}
Under Assumption \eqref{assumption 1}, matrix $A-W$ is invertible. Hence, the NE of the game can be obtained as follows:

\begin{equation}
    q^*(t)=(A-W)^{-1}(B-\Phi^*(t)).
\end{equation}

\end{proof}

\subsection{Price and maintenance scheduling of the leader}
As we mentioned above, in the proposed model, we assume that the supplier can predict the strategies of followers. Based on this prediction, the supplier optimizes its strategy. Moreover, the supplier acts first, in the sense that it sends its strategy to the followers, which then decide about their actions. We assume that the pricing is performed at the level of an individual customer. Customers who are more connected within the social network (and are, therefore, impacted by the network) can than be better targeted by the supplier. The strategy of the supplier comprises the price of the network and the maintenance scheduling of the manufacturing units. Hence, in the following, we obtain the optimal price and maintenance scheduling of the supplier relying on the  fact that it knows the NE of the game \eqref{eq:NE} among the customers. 

Let us define the matrix $(A-W)^{-1}$ as follows:
\begin{equation}
    (A-W)^{-1}=(R_{il})_{i,l\in\mathcal{N}}.
    \label{eq:matrix}
\end{equation}
In view of equality \eqref{eq:NE} and using \eqref{eq:matrix}, the optimization problem of supplier \eqref{eq:main problem1} can be defined as the following mixed-integer quadratic programming problem:
\begin{equation}
\begin{aligned}
    &\max\limits_{{\phi},{x},{s},{y}}\sum\limits_{t\in\mathcal{T}}\phi^\top(t)(A-W)^{-1}(B-\Phi(t))-\sum\limits_{j\in\mathcal{J}}\sum\limits_{t\in\mathcal{T}}c_{j}x_{j}(t)\\
    &\text{S.b.}\quad \text{C}'_{1},\quad \text{C}'_{2}\quad ,\quad \text{C}'_{3},\quad\text{C}'_{4},\quad \text{C}'_{5},\\
    &\qquad\quad\text{C}''_{6}:\quad\sum\limits_{l\in\mathcal{N}}R_{il}(b_{l}-\phi_{l})\leq{\sum\limits_{j\in\mathcal{J}}(1-x_{j}(t))q_{j,max}},\\
    &\quad \qquad \text{C}''_{7}:\quad (A-W)^{-1}(B-\Phi(t))\geq{0},\quad \phi(t)\geq{0}.
    \end{aligned}
    \label{eq:main problem2}
\end{equation}
By solving \eqref{eq:main problem2}, the optimal price $\phi^*_{i}$, $i\in\mathcal{N}$, can be obtained and sent to the customers. Based on the price, they can obtain their NE using \eqref{eq:NE}.
\section{Simulation results}
\label{sec:simulation}
In this section, we implement the proposed method on a manufacturing system where the supplier has five manufacturing units and aims to obtain their maintenance scheduling for $30$ days. We assume a normal distribution for the deterioration threshold value. The mean and the standard deviation of the deterioration state thresholds and the maintenance costs of each of the manufacturing units are displayed in Table \ref{tab:thereshold}.

\begin{table}
  \caption{Deterioration and maintenance parameters of the manufacturing units}
  \centering
 \begin{tabular}{c | c c c} 
Unit &  Mean of deterioration  & Standard deviation of deterioration & Maintenance cost  \\
&  state thresholds $(S_{2,j})$ &  state thresholds $(S_{2,j})$ &  $(c_{j})$ $[MU]$ \\
 \hline
 1 &  12& 1.4 &20.48 \\ 
 \hline
 2 & 10  & 3.2 &21.39 \\
 \hline
 3 & 11.2& 2.5 &22.73\\
 \hline
 4 &  9.4& 1.1& 24.78\\
 \hline
 5 & 11.8 &2.1 & 24.82\\
\end{tabular}
  \label{tab:thereshold}
\end{table}

Moreover, we assume $10$ customers who are connected through a social network graph. 
The profile of the customers' demand is depicted in Figure \ref{demand1}.
\begin{figure}
	\centering
	\includegraphics[width=5.5in]{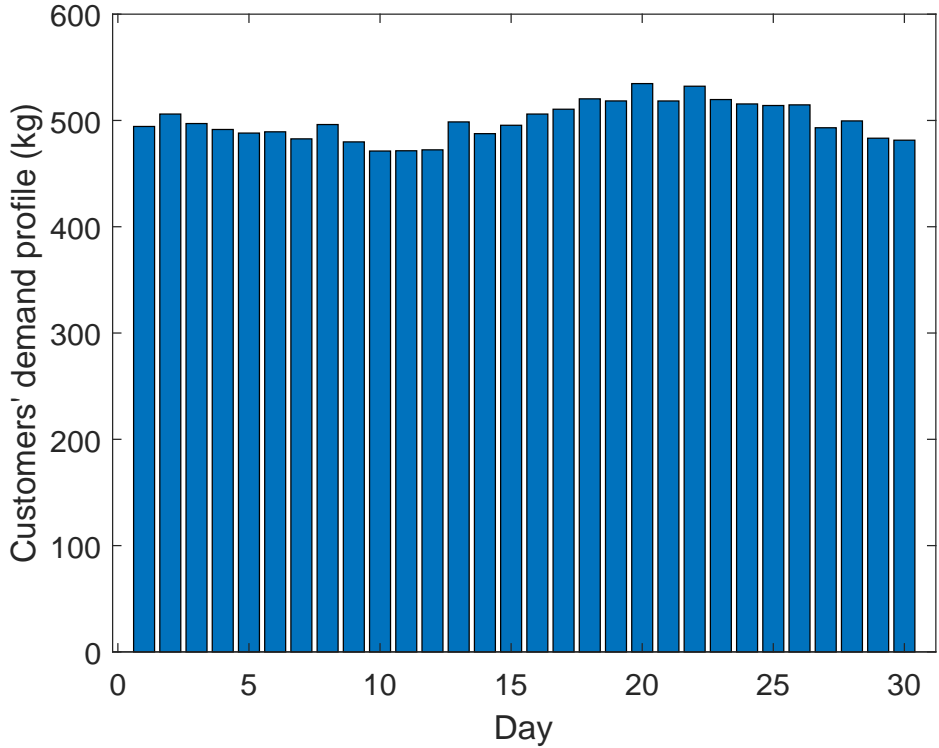}
	\caption{The total demand of customers during $30$ days.}
	\label{demand1}
\end{figure}
We consider the following graph weights among the customers in the network:
\begin{equation}
W=
\begin{pmatrix}
0&0.54&0.82&0.28&0.11&0.41&0.16&0.05&0.76&0.86\\
0.54&0&0.79&0&0&0&0&0&0&0\\0.82&0.79&0&0.36 &0 &0&0&0&0&0\\
0.28&0 &     0.36&0& 0 &0&0 &0&0\\                     0                            0.11&0&0&0.08&0&0.62&0&0&0&0\\0.41&0&0&0&                           0.62&0&0&0&0\\0.16&0&0&0&0&0.63&0&0.75&0&0\\0.05&0&0&0&0&0&0.75&0&0.54&0\\0.76&0&0&0&0&0&0&0.54&0&0.25\\0.86&0&0&0&0&0&0&0&0.25&0
\end{pmatrix}
\end{equation}
The weight matrix shows that customer $1$ is connected to all other customers.  Hence, his/her decisions have a high impact on the decisions of other customers. Figure \ref{consumption} shows the consumption level of the customers. As we can see, the consumption level of customer $1$ is higher than that of other customers since the supplier offers a lower price to this customer to encourage him/her to consume more goods which has a high impact on other customers' strategies. 
\begin{figure}
	\centering
	\includegraphics[width=5.5in]{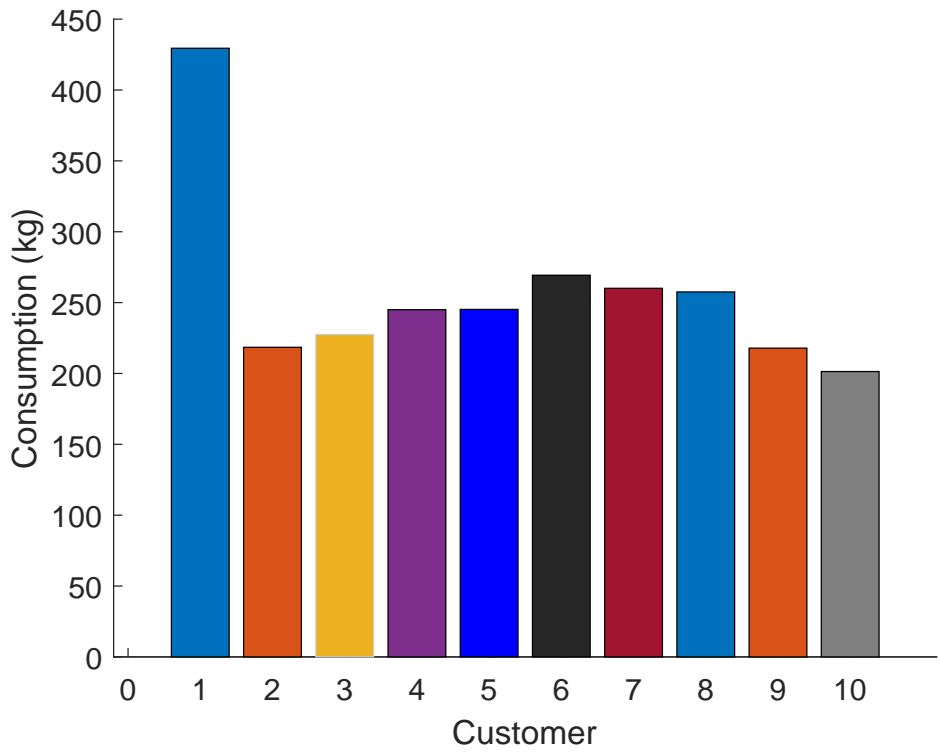}
	\caption{The total consumption of customers during $30$ days.}
	\label{consumption}
\end{figure}
The network price and maintenance scheduling of the units are shown in Figures \ref{price} and \ref{maintenance}.
\begin{figure}
	\centering
	\includegraphics[width=5.5in]{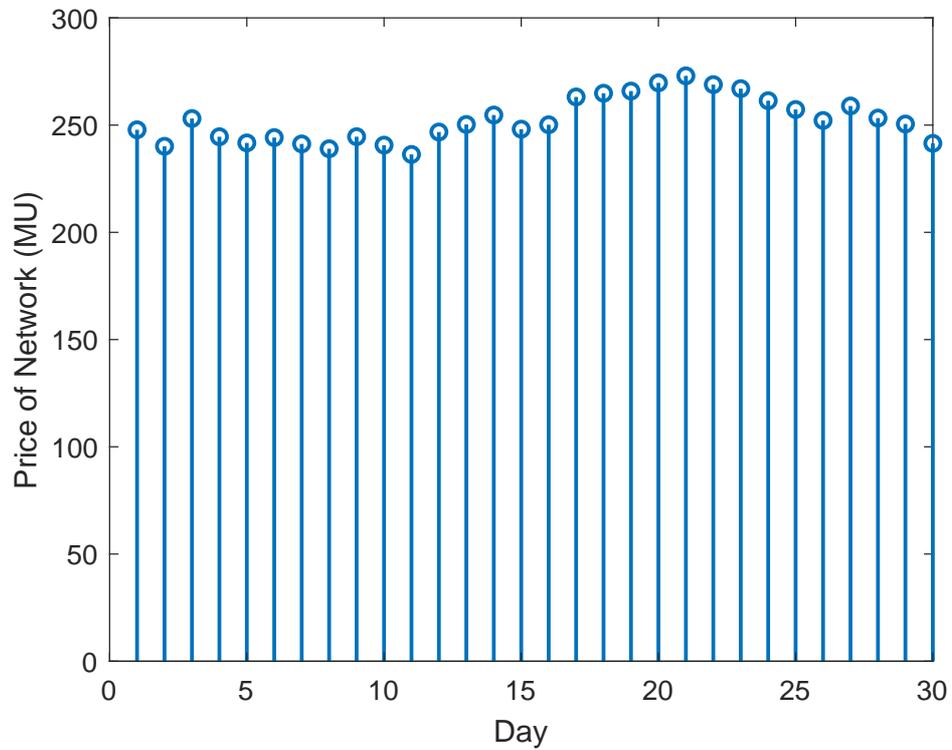}
	\caption{The profile of the network price during $30$ days.}
	\label{price}
\end{figure}
As we can see from Figure \ref{price} and \ref{maintenance}, the manufacturing units perform maintenance more frequently when the price of the network is low (in the first $15$ days)  compared to the time periods when the network price is high (between day $15$ and $25$). 
\begin{figure}
	\centering
	\includegraphics[width=5.5in]{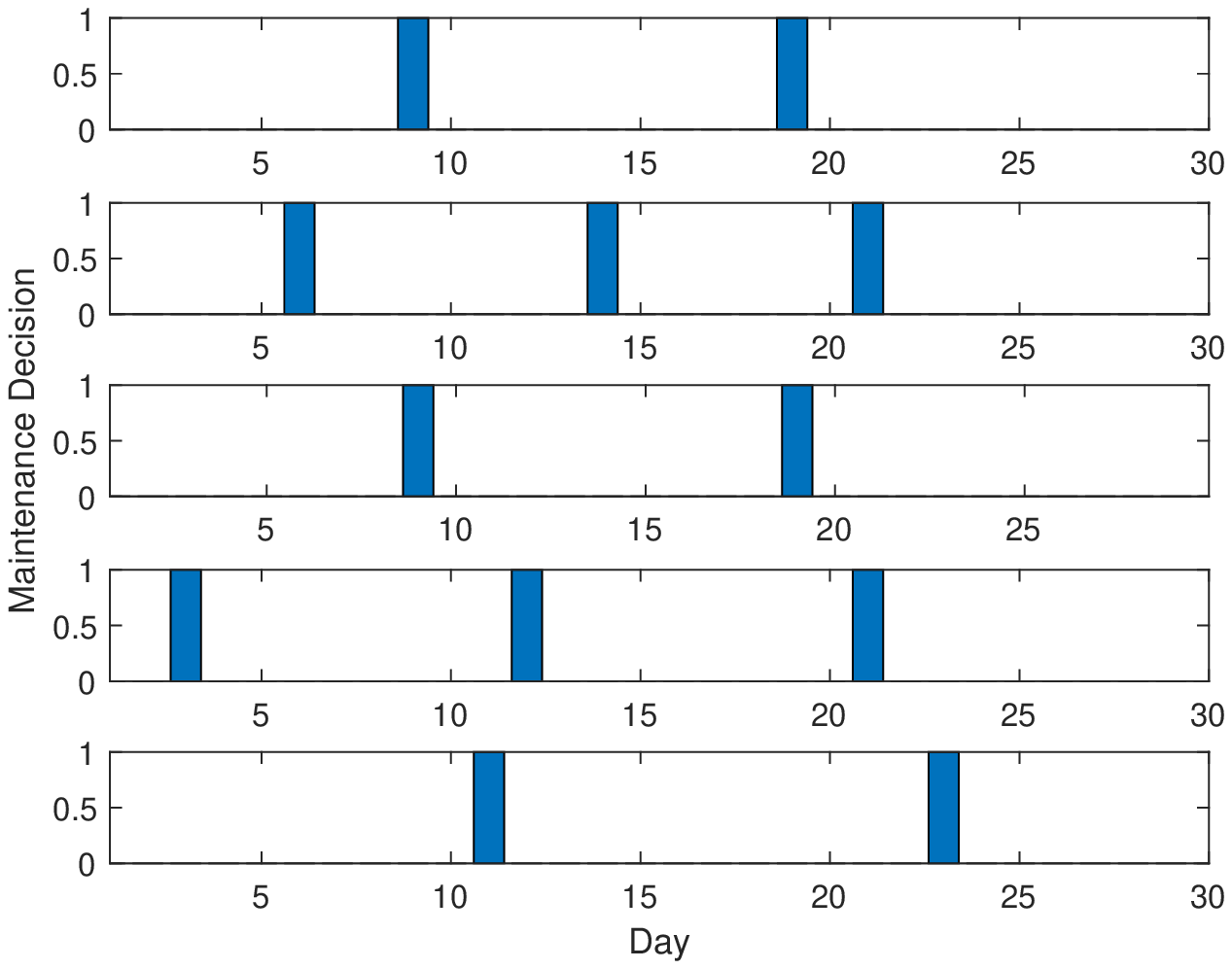}
	\caption{Maintenance scheduling of all supplier's manufacturing units during $30$ days.}
	\label{maintenance}
\end{figure}
The deterioration state of the units is shown in Figure \ref{deterioration}. The figure shows that the units perform maintenance before their failure time 
\begin{figure}
	\centering
	\includegraphics[width=5.5in]{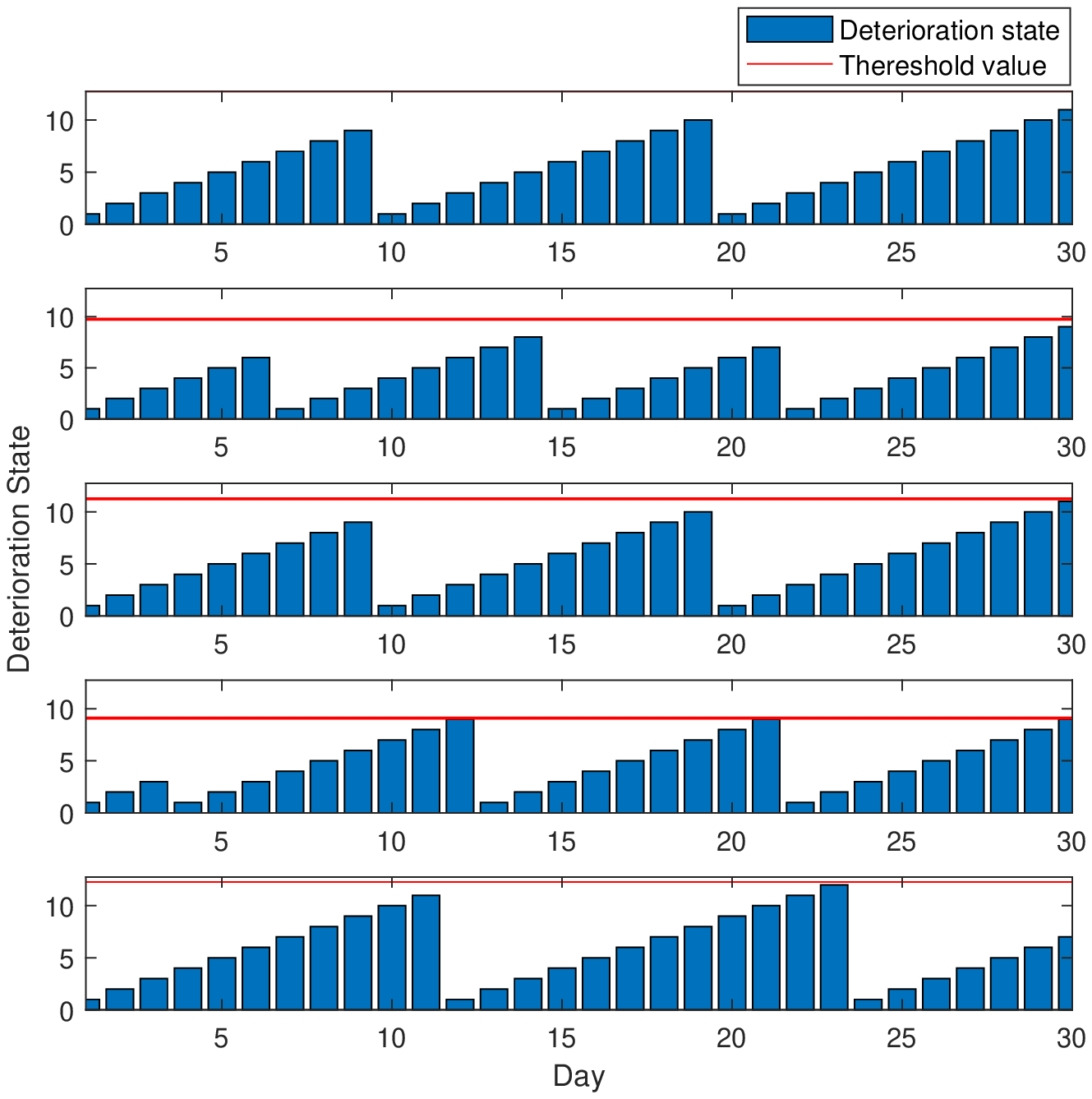}
	\caption{Deterioration state of all supplier's units during $30$ days.}
	\label{deterioration}
\end{figure}
\subsection{The effect of network structure}
In this section, we investigate the effect of the social network structure on the supplier's objective function. In case that the weights of the network are not known to the supplier, the supplier predicts the NE of the network as follows:
\begin{equation}
    q^*(t)=A^{-1}(B-\phi^*(t))
\end{equation}
Then, the price and maintenance scheduling are obtained as follows:

\begin{equation}
\begin{aligned}
    &\max\limits_{{\phi},{x},{s},{y}}\sum\limits_{t\in\mathcal{T}}\phi^\top(t)(A)^{-1}(B-\Phi(t))-\sum\limits_{j\in\mathcal{J}}\sum\limits_{t\in\mathcal{T}}c_{j}x_{j}(t)\\
    &\text{S.b.}\quad \text{C}'_{1},\quad \text{C}'_{2}\quad ,\quad \text{C}'_{3},\quad\text{C}'_{4},\quad \text{C}'_{5},\\
    &\qquad\quad\text{C}''_{6}:\quad\sum\limits_{l\in\mathcal{N}}R_{il}(b_{l}-\phi_{l})\leq{\sum\limits_{j\in\mathcal{J}}(1-x_{j}(t))q_{j,max}},\\
    &\quad \qquad \text{C}''_{7}:\quad (A)^{-1}(B-\Phi(t))\geq{0},\quad \phi(t)\geq{0}.
    \end{aligned}
    \label{eq:main problem3}
\end{equation}
In our case study, the profit that the supplier can obtain by solving \eqref{eq:main problem2} (knowing the network structure) for the entire duration of the period over $30$ days is $1.2893e+05$ $MU$ (monetary unit) and the supplier's objective function without knowing the network structure \eqref{eq:main problem3} is $5.3622e+04$ $MU$. Hence, by knowing the network structure, a supplier can obtain the strategy that results in a $7.5308e+04$ $MU$ higher profit.

\subsection{Comparison of the proposed method to the baseline maintenance decision}
The result of the reward function is compared to the baseline solution where the units perform maintenance at the deterioration threshold $\min\limits_{k=1,\cdots,K}S^{k}_{2,j}$, $j\in\mathcal{J}$. This baseline decision makes sure that the units can perform maintenance before they fail. The maintenance decision and the degradation cost are shown in Figures \ref{maintenance decision_baseline} and \ref{degredation cost_baseline}.

\begin{figure}
	\centering
	\includegraphics[width=5.5in]{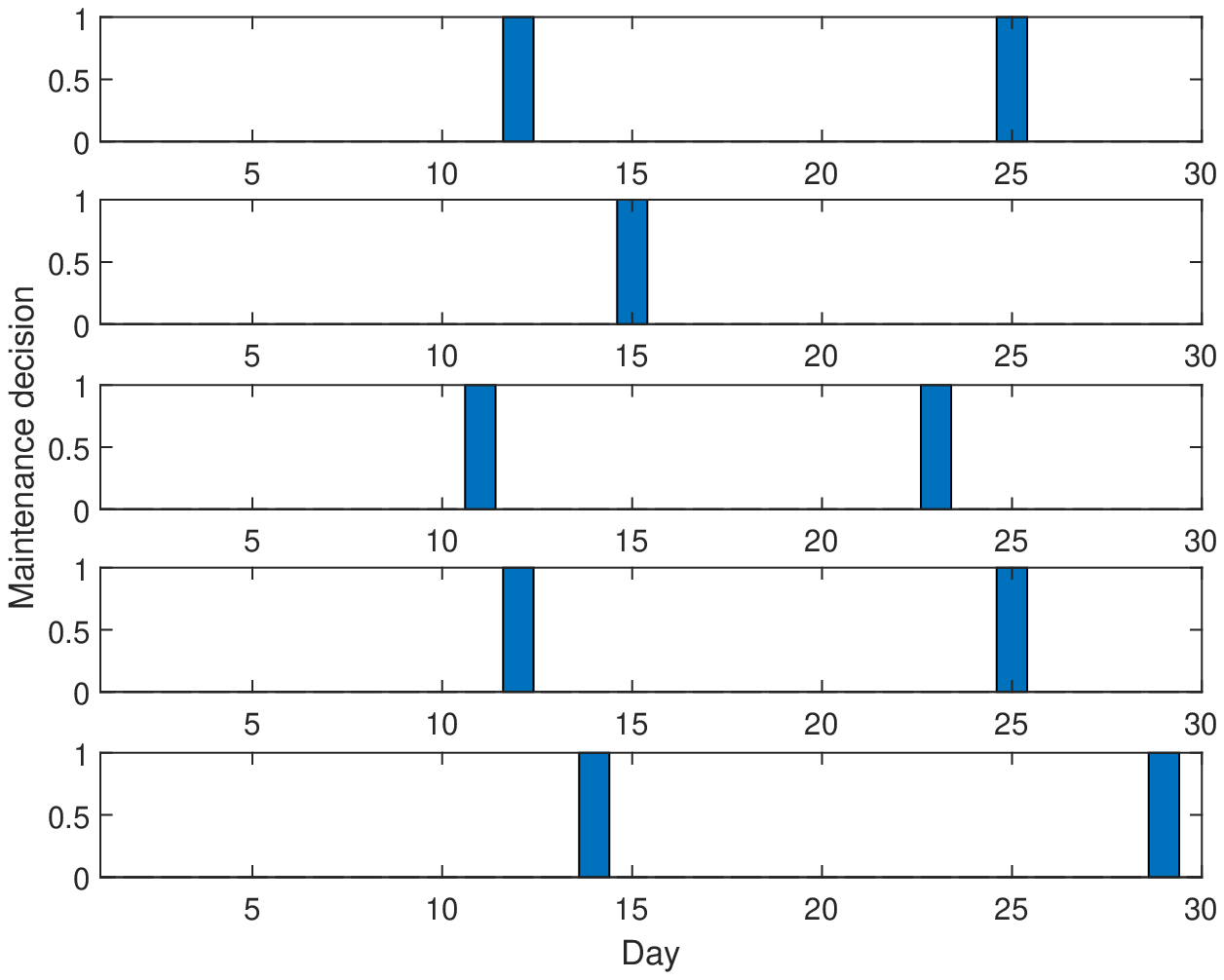}
	\caption{Baseline maintenance scheduling of all supplier's manufacturing units during $30$ days.}
	\label{maintenance decision_baseline}
\end{figure}

\begin{figure}
	\centering
	\includegraphics[width=5.5in]{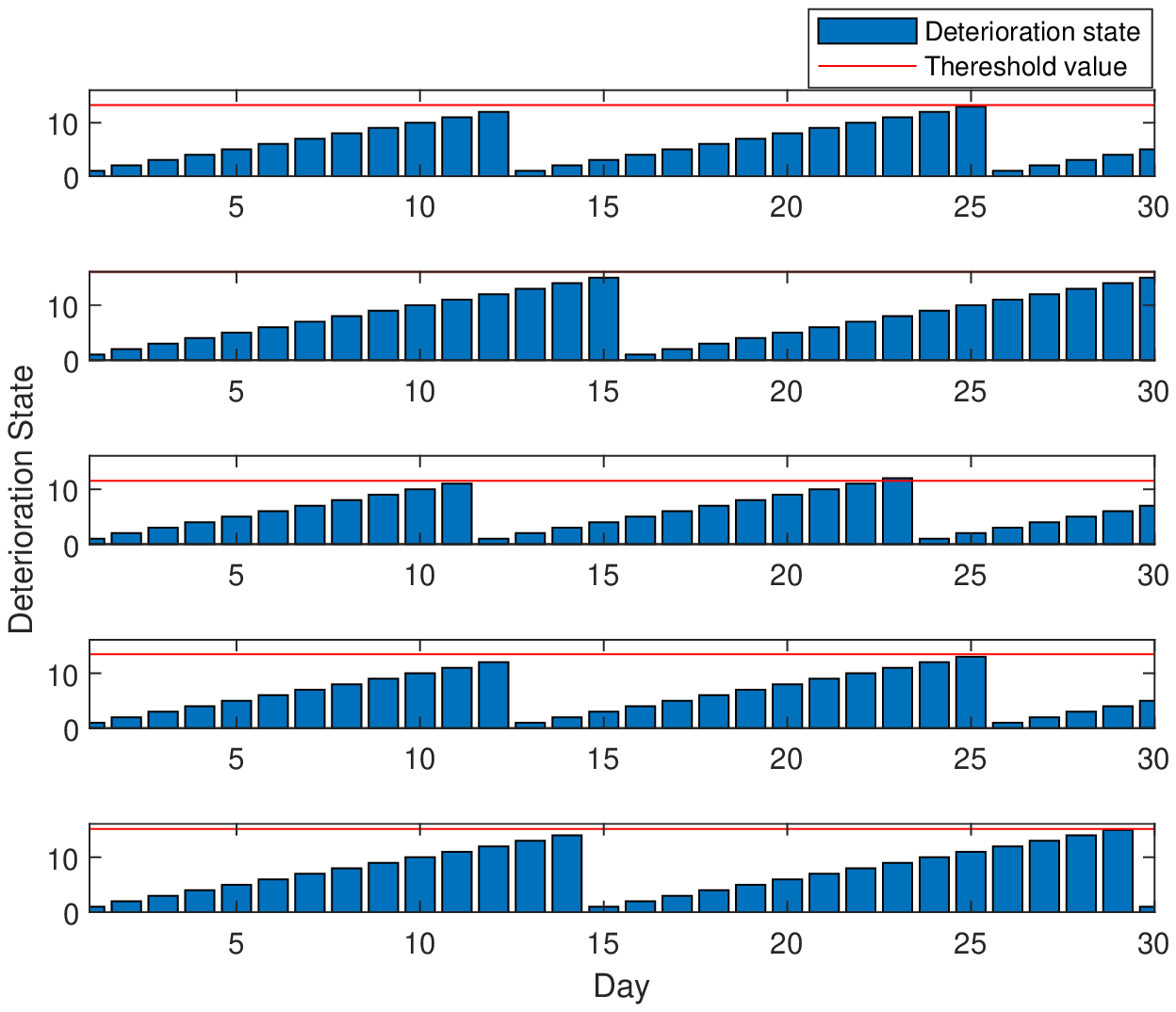}
	\caption{Deterioration state of all supplier's units for baseline maintenance decision during $30$ days .}
	\label{degredation cost_baseline}
\end{figure}
When the units decide to perform maintenance when the threshold is reached, the profit of the supplier is 
$7.4919e+04$ $MU$, which is $5.4011$ $MU$ less than the profit when they obtain their maintenance decision using \eqref{eq:main problem1}. In this case, the maintenance scheduling of units does not depend on the price of the network but solely on the deterioration state. 
\section{Conclusion}
\label{sec:conclusion}
In this paper, we address the problem of maintenance scheduling for manufacturing units of goods with positive externalities by concurrently optimizing the pricing and the maintenance schedule.  The customers are connected through a social network graph. In order to solve the problem, we propose a bi-level optimization approach where at the first level, the supplier obtains its maintenance scheduling and the price at the level of individual customers. At the second level, the customers obtain their strategies based on the price offered by the supplier. In order to solve the bi-level optimization problem, we propose a solution based on a leader-multiple-followers game where the customers are the followers and the supplier acts as a leader that optimizes its decisions based on the predicted strategies of the customers. The proposed approach can be extended and implemented to the case of more complex network with more agents. The numerical results of the case study show that when the network structure is known to the supplier, the obtained profit of the supplier is higher compared to the case that the connections between the customers in the network are not known. Moreover, we demonstrated that obtaining the maintenance decision using the proposed bi-level optimization method leads to more profit for the supplier than base-line solution.

As future research, multiple suppliers can be considered in the proposed framework where the strategies of the suppliers are related to each other and they also play a game. Moreover, one can consider uncertainties in the RUL prediction and the customers' models. This will enable to develop a robust optimization method which can take the uncertainties into account.

\section*{Acknowledgement}
\vspace{-3pt}
The contribution of Olga Fink was funded by the Swiss National Science Foundation (SNSF) Grant no. PP00P2\_176878.
\vspace{-3pt}

\bibliographystyle{elsarticle-num}
\bibliography{bib_items}

\end{document}